\def\year{2020}\relax
\documentclass[letterpaper]{article} 
\usepackage{aaai20}  
\usepackage{times}  
\usepackage{helvet} 
\usepackage{courier}  
\usepackage[hyphens]{url}  
\usepackage{graphicx} 
\usepackage{graphicx}  
\usepackage[utf8]{inputenc} 
\usepackage{url}            
\usepackage{booktabs}       
\usepackage{amsfonts}       
\usepackage{nicefrac}       
\usepackage{microtype}      
\usepackage{color}
\usepackage[inline]{enumitem}

\usepackage[final]{changes}

\usepackage{algorithm} 
\usepackage{algorithmic} 

\usepackage{amsfonts}
\usepackage{lipsum}
\usepackage{graphicx}
\usepackage{subfigure}
\usepackage{amstext}
\usepackage{amsmath}
\usepackage{mathrsfs}
\usepackage{enumitem}
\usepackage{amsthm}
\usepackage{amsmath,amsfonts,amssymb,amsthm,epsfig,epstopdf,url,array}
\usepackage{todonotes}
\usepackage{multirow}
\usepackage{enumitem}
\usepackage{graphicx}

\setitemize{noitemsep,topsep=0pt,parsep=0pt,partopsep=0pt}

\newtheorem{theorem}{Theorem}
\newtheorem{lemma}[theorem]{Lemma}

\theoremstyle{definition}

\newtheorem{corollary}[theorem]{Corollary}

\author{
\parbox{\linewidth}{\centering
Heng Chang\thanks{Work done while interning at Tencent AI Lab.},\textsuperscript{\rm 1}
Yu Rong,\textsuperscript{\rm 2}
Tingyang Xu,\textsuperscript{\rm 2}
Wenbing Huang,\textsuperscript{\rm 2}
Honglei Zhang,\textsuperscript{\rm 2} \\
Peng Cui,\textsuperscript{\rm 3}
Wenwu Zhu\thanks{Wenwu Zhu is the corresponding author.},\textsuperscript{\rm 1,3}
Junzhou Huang\textsuperscript{\rm 2} } \\
\textsuperscript{\rm 1}Tsinghua-Berkeley Shenzhen Institute, Tsinghua University, China \\
\textsuperscript{\rm 2}Tencent AI Lab, China \\
\textsuperscript{\rm 3}Department of Computer Science and Technology, Tsinghua University, China \\
changh17@mails.tsinghua.edu.cn, yu.rong@hotmail.com, Tingyangxu@tencent.com, hwenbing@126.com, \\ zhanghonglei@gatech.edu, cuip@tsinghua.edu.cn, wwzhu@tsinghua.edu.cn, jzhuang@uta.edu
}

\title{A Restricted Black-box Adversarial Framework Towards Attacking Graph Embedding Models}

\begin{document}

\maketitle

\begin{abstract}
With the great success of graph embedding model on both academic and industry area, the robustness of graph embedding against adversarial attack inevitably becomes a central problem in graph learning domain. Regardless of the fruitful progress, most of the current works perform the attack in a white-box fashion: they need to access the model predictions and labels to construct their adversarial loss. However, the inaccessibility of model predictions in real systems makes the white-box attack impractical to real graph learning system. 
This paper promotes current frameworks in a more general and flexible sense -- we demand to attack various kinds of graph embedding model with black-box driven. To this end, we begin by investigating the theoretical connections between graph signal processing and graph embedding models in a principled way and formulate the graph embedding model as a general graph signal process with corresponding graph filter. As such, a generalized adversarial attacker: \textit{GF-Attack} is constructed by the graph filter and feature matrix. Instead of accessing any knowledge of the target classifiers used in graph embedding, \textit{GF-Attack} performs the attack only on the graph filter in a black-box attack fashion. To validate the generalization of \textit{GF-Attack}, we construct the attacker on four popular graph embedding models. Extensive experimental results validate the effectiveness of our attacker on several benchmark datasets.
Particularly by using our attack, even small graph perturbations like one-edge flip is able to consistently make a strong attack in performance to different graph embedding models.

\end{abstract}

\section{Introduction}

\begin{figure*}[htb]
\centering
\includegraphics [width=0.9\textwidth]{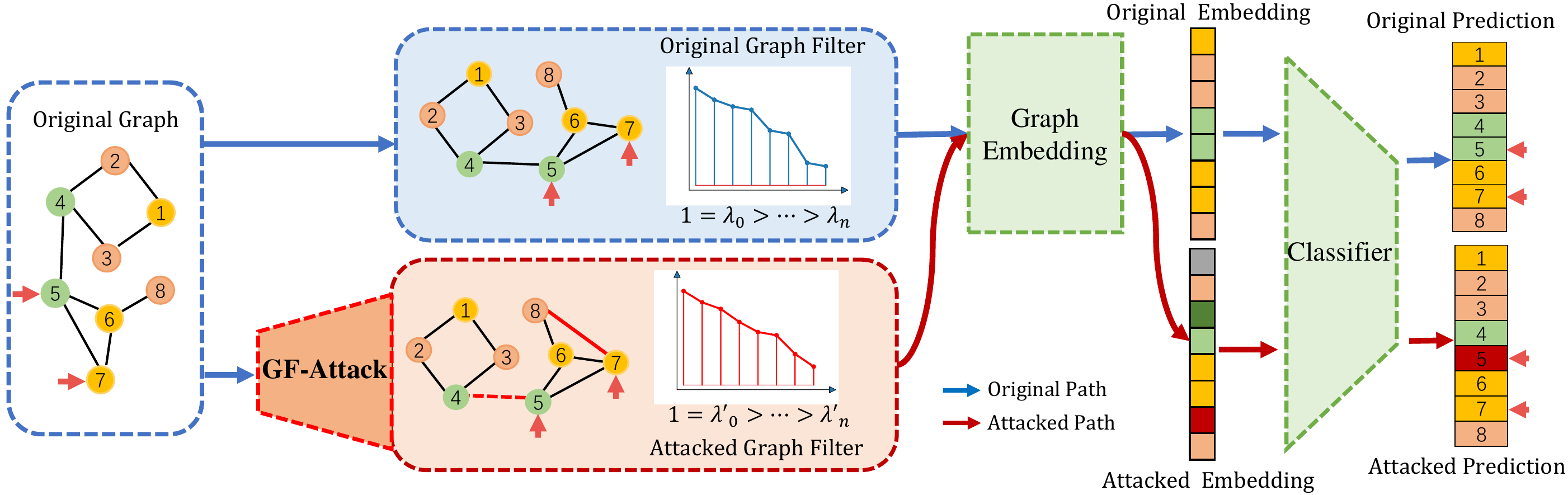}

\caption{The overview of whole attack procedure of \textit{GF-Attack}. Given target vertices $5$ and $7$, \textit{GF-Attack} aims to misclassify them by attacking the graph filter and producing adversarial edges (edge $e_{45}$ deleted and edge $e_{78}$ added ) on graph structure. The common graph embedding block refers to the general target GNN model and can be any kind of potential GNN models, illustrating the flexibility and extensibility of GF-Attack. In this vein, GF-Attack would not change the target embedding model.}
\label{fig.attackoverview}
\end{figure*}

Graph embedding models ~\cite{scarselli2009GNN,cui2018survey}, which elaborate the expressive power of deep learning on graph-structure data, have achieved promising success in various domains, such as predicting properties over molecules~\cite{duvenaud2015convolutional},  biology analysis~\cite{Hamilton2017Inductive}, financial surveillance~\cite{paranjape2017motifs} and structural role classification~\cite{tu2018deep}.
Given the increasing popularity and success of these methods, a bunch of recent works have posed the risk of graph embedding models against adversarial attacks, just like what the researchers are anxious for convolutional neural networks~\cite{akhtar2018threat}. A strand of research works~\cite{ICML2018Adversarial,KDD2018Adversarial,icml2019adversarial} have already shown that various kinds of graph embedding methods, including Graph Convolutional Networks, DeepWalk, etc., are vulnerable to adversarial attacks. 
Undoubtedly, the potential attacking risk is rising for modern graph learning systems. For instance, by sophisticated constructed social bots and following connections, it's possible to fool the recommendation system equipped with graph embedding models to give wrong recommendations. 

Regarding the amount of information from both target model and data required for the generation of adversarial examples, all graph adversarial attackers fall into three categories (arranged in an ascending order of difficulties):
\begin{itemize}
    \item White-box Attack (\textbf{WBA}): the attacker can access any information, namely, the training input (e.g., adjacency matrix and feature matrix), the label, the model parameters, the predictions, etc.
    \item Practical White-box Attack (\textbf{PWA}): the attacker can any information except the model parameters. 
    \item Restrict Black-box Attack (\textbf{RBA}): the attacker can only access the training input and limited knowledge of the model. The access of parameters, labels and predictions is prohibited. 
\end{itemize}

Despite the fruitful results \cite{sun2018adversarial,KDD2018Adversarial,ICLR2019Meta} which absorb ingredients from exiting adversarial methods on convolutional neural networks, obtained in attacking graph embeddings under both WBA and PWA setting, however, the target model parameter as well as the labels and predictions are seldom accessible in real-life applications. In the other words, the WBA and PWA attackers are almost impossible to perform a threatening attack to real systems. Meanwhile, current RBA attackers are either reinforcement learning based \cite{ICML2018Adversarial}, which has low computational efficiency and is limited to edge deletion, or derived merely only from the structure information without considering the feature information~\cite{icml2019adversarial}. Therefore, how to perform the effective adversarial attack toward graph embedding model relying on the training input, a.k.a., RBA setting, is still more challenging yet meaningful in practice.

The core task of the adversarial attack on graph embedding model is to damage the quality of output embeddings to harm the performance of downstream tasks within the manipulated features or graph structures, i.e., vertex or edge insertion/deletion. Namely, finding the embedding quality measure to evaluate the damage of embedding quality is vital. For the WBA and PWA attackers, they have enough information to construct this quality measure, such as the loss function of the target model. In this vein, the attack can be performed by simply maximize the loss function reversely, either by gradient ascent \cite{ICML2018Adversarial} or a surrogate model \cite{KDD2018Adversarial,ICLR2019Meta} given the known labels. However, the RBA attacker can not employ the limited information to recover the loss function of the target model, even constructing a surrogate model is impossible. In a nutshell, the biggest challenge of the RBA attacker is: how to figure out the goal of the target model barely by the training input.

In this paper,  we try to understand the graph embedding model from a new perspective and propose an attack framework: \textit{GF-Attack}, which can perform adversarial attack on various kinds of graph embedding models. Specifically, we formulate the graph embedding model as a general graph signal processing with corresponding graph filter which can be computed by  the input adjacency matrix. Therefore, we employ the graph filter as well as feature matrix to construct the embedding quality measure as a $T$-rank approximation problem. In this vein, instead of attacking the loss function, we aim to attack the graph filter of given models. It enables \textit{GF-Attack} to perform attack in a restrict black-box fashion. Furthermore, 
by evaluating this $T$-rank approximation problem, \textit{GF-Attack} is capable to perform the adversarial attack on any graph embedding models which can be formulate to a general graph signal processing. Meanwhile, we give the quality measure construction for four popular graph embedding models (GCN, SGC, DeepWalk, LINE). Figure~\ref{fig.attackoverview} provides the overview of whole attack procedure of \textit{GF-Attack}. Empirical results show that our general attacking method is able to effectively propose adversarial attacks to popular unsupervised/semi-supervised graph embedding models on real-world datasets without access to the classifier.

\section{Related work}\label{sec.related}
For explanation of graph embedding models, \cite{xu2018how} and \cite{WSDM2018NetworkEmbedding} show some insights on the understanding of Graph Convolutional Networks and sampling-based graph embedding, respectively. However, they focus on proposing new graph embedding frameworks in each type of methods rather than building up a theoretical connection.

Only recently adversarial attacks on deep learning for graphs have drawn unprecedented attention from researchers.
\mbox{\cite{ICML2018Adversarial}} exploits a reinforcement learning based framework under RBA setting. However, they restrict their attacks on edge deletions only for node classification, and do not evaluate the transferability. \cite{KDD2018Adversarial} proposes attacks based on a surrogate model and they can do both edge insertion/deletion in contrast to \cite{ICML2018Adversarial}. But their method utilizes additional information from labels, which is under PWA setting.
Further, \cite{ICLR2019Meta} utilizes meta-gradients to conduct attacks under black-box setting by assuming the attacker uses a surrogate model same as \cite{KDD2018Adversarial}. Their performance highly depends on the assumption of the surrogate model, and also requires label information. Moreover, they focus on the global attack setting.
\cite{xu2019topology} also proposes a gradient-based method under WBA setting and overcomes the difficulty brought by discrete graph structure data.
\cite{icml2019adversarial} considers a different adversarial attack task on vertex embeddings under RBA setting. Inspired by \cite{WSDM2018NetworkEmbedding}, they maximize the loss obtained by DeepWalk with matrix perturbation theory while only consider the information from adjacent matrix.
In contrast, we focus on semi-supervised learning on node classification combined with features. Remarkably, despite all above-introduced works except \cite{ICML2018Adversarial} show the existence of transferability in graph embedding methods by experiments, they all lack theoretical analysis on the implicit connection. In this work, for the first time, we theoretically connect different kinds of graph embedding models and propose a general optimization problem from parametric graph signal processing. An effective algorithm is developed afterwards under RBA setting.

\section{Preliminary}\label{Background}
Let $G(\mathcal{V},\mathcal{E})$ be an attributed graph, where $\mathcal{V}$ is a vertex set with size $n = |\mathcal{V}|$ and $\mathcal{E}$ is an edge set.  Denote $A \in \{0,1\}^{n \times n}$ as an adjacent matrix containing information of edge connections and $X \in \mathbb{R}^{n \times l}$ as a feature matrix with dimension $l$\deleted[id==RR]{ for vertices}. $D_{ii} = \sum_{j}A_{ij}$ refers the degree matrix. $\text{vol}(G) = \sum_{i}\sum_{j}A_{ij} = \sum_{i}D_{ii}$ denotes the volume of $G$. For consistency, we denote the perturbed adjacent matrix as $A'$ and the normalized adjacent matrix as $\hat{A} = D^{-\frac{1}{2}}AD^{-\frac{1}{2}}$. Symmetric normalized Laplacian and random walk normalized Laplacian are referred as $L^{sym} = I_n - D^{-\frac{1}{2}}AD^{-\frac{1}{2}}$ and $L^{rw} = I_n - D^{-1}A$, respectively.

Given a graph embedding model $\mathscr{M}_\Theta$ parameterized by $\Theta$ and a graph $G(\mathcal{V}, \mathcal{E})$,  the adversarial attack on graph aims to perturb the learned vertex representation $Z = \mathscr{M}_{\Theta}(A, X)$ to damage the performance of the downstream learning tasks. There are three components in graphs that can be attacked as targets:
\begin{itemize}
  \item Attack on $\mathcal{V}$: Add/delete vertices in graphs. This operation may change the dimension of the adjacency matrix $A$.
  \item Attack on $A$: Add/delete edges in graphs. This operation would lead to the changes of entries in the adjacency matrix $A$. This kind of attack is also known as \emph{structural attack}.
  \item Attack on $X$: Modify the attributes attached on vertices.  
\end{itemize}
Here, we mainly focus on adversarial attacks on graph structure $A$, since attacking $A$ is more practical than others in real applications \cite{CIKM2012Gelling}.

\subsection{Adversarial Attack Definition}
Formally, given a fixed budget $\beta$ indicating that the attacker is only allowed to modify $2\beta$ entries in $A$ (undirected), the adversarial attack on a graph embedding model $\mathscr{M}_\Theta$ can be formulated as \cite{icml2019adversarial}:
\begin{align}\label{equ.problemdef}
    \arg\max\limits_{A'}  & \,\, \mathscr{L}(A', Z)\\ 
    \notag\text{s.t.}\; &\notag Z = \mathscr{M}_{\Theta}(A', X), \\ \Theta^{*} &= \arg\min_{\Theta}\mathcal{L}(\Theta; A', X),
                        \notag\| A' - A\| = 2\beta,
\end{align}
where $Z$ is the embedding output of the model $\mathscr{M}_{\Theta}$ and $\mathcal{L}(\cdot, \cdot)$ is the loss function minimized by $\Theta$. $\mathscr{L}(A', Z)$ is defined as the loss measuring the attack damage on output embeddings, lower loss corresponds to higher quality. For the WBA, $\mathscr{L}(A', Z)$ can be defined by the minimization of the target loss, i.e.,  $\mathscr{L}(A', Z) = \inf\limits_{\Theta}\mathcal{L} (A', Z)$. This is a bi-level optimization problem if we need to re-train the model during attack. Here we consider a more practical scenario: $\Theta^{*} = \arg\min_{\Theta}\mathcal{L}(\Theta; A, X)$ are learned on the clean graph and remains unchanged during attack.

\section{Methodologies}\label{sec.GSPG}
Graph Signal Processing (GSP) focuses on analyzing and processing data points whose relations are modeled as graph \cite{shuman2013GSP,ortega2018graph}.  Similar to Discrete Signal Processing, these data points can be treated as \emph{signals}. Thus the definition of \emph{graph signal} is a mapping from vertex set $\mathcal{V}$ to real numbers $\mathbf{x}: \mathcal{V} \rightarrow \mathbb{R}$. In this sense, the feature matrix $X$ can be treated as graph signals with $l$ channels.
From the perspective of GSP, we can formulate graph embedding model $\mathscr{M}:(A, X) \to \mathbb{R}^{n \times d}$ as the generalization of signal processing. Namely, A graph embedding model can be treated as producing the new graph signals according to graph filter $\mathscr{H}$ together with feature transformation: 
\begin{equation}
    \tilde{X} = \mathscr{H}(X), \,
    X' =  \sigma(\tilde{X}\Theta),
    \label{equ.GF-Attack}
\end{equation}
where $\mathscr{H}$ denotes a graph signal filter, $\sigma(\cdot)$ denotes the activation function of neural networks, and $\Theta \in \mathbb{R}^{l \times l'}$ denotes a convolution filter from $l$ input channels to $l'$ output channels. $\mathscr{H}$ can be constructed by a polynomial function $h(x)=\sum_{i=0}^La_ix^i \in \mathbb{R}^{n \times n}$ with graph-shift filter $S$, i.e., $\tilde{X} =h(S)X$. Here, the graph-shift filter $S$ reflects the locality property of graphs, i.e., it represents a linear transformation of the signals of one vertex and its neighbors. It's the basic building blocks to construct $\mathscr{H}$. Some common choices of $\mathscr{H}$ include the adjacency matrix $A$ and the Laplacian $L=D - A$. We call this general model \emph{Graph Filter Attack~(GF-Attack)}. \textit{GF-Attack} introduces the trainable weight matrix $\Theta$ to enable stronger expressiveness which can fuse the structural and non-structural information. 

\subsection{Embedding Quality Measure $\mathscr{L}(A', Z)$ of \textit{GF-Attack}}
According to~\eqref{equ.GF-Attack}, in order to avoid accessing the target model parameter $\Theta$, we can construct the restricted black-box attack loss $\mathscr{L} (A', Z)$ by attacking the graph filter $\mathscr{H}$. 
Recent works \cite{yang2015network,nar2019cross} demonstrate that the output embeddings of graph embedding models can have very low-rank property.
Since our goal is to damage the quality of output embedding $Z$, we establish the general optimization problem accordingly as a $T$-rank approximation problem inspired from \cite{WSDM2018NetworkEmbedding}:
\begin{align}
\mathscr{L}(A', Z) =  \|h(S')X - h(S')_TX \|_F^2,
\end{align}
where $h(S')$ is the polynomial graph filter, $S'$ is the graph shift filter constructed from the perturbed adjacency matrix $A'$. $ h(S')_T$ is the $T$-rank approximation of $h(S')$. According to low-rank approximation, $\mathscr{L}(A', Z)$ can be rewritten as:
\begin{equation}
\small{
\mathscr{L}(A', Z) = \| \sum_{i = T + 1}^{n} \lambda'_{i} \mathbf{u}_{i}\mathbf{u}_{i}^{T}X \|_{F}
\leq \sum_{i = T + 1}^{n} {\lambda'_{i}}^{2} \cdot \sum_{i = T + 1}^{n} \|\mathbf{u}_{i}^{T}X\|_2^2,
\label{equ.loss}
}
\end{equation}
where $n$ is the number of vertices. $h(S')=  U\Lambda U^{\text{T}}$ is the eigen-decomposition of the graph filter $h(S')$. $h(S')$ is a symmetric matrix. $\Lambda= diag(\lambda_1,\cdots,\lambda_n)$, $U= [\mathbf{u}_{1}^{T},\cdots,\mathbf{u}_{n}^{T}]$ are the eigenvalue and eigenvector of graph filter $\mathscr{H}$, respectively, in order of $ \lambda_{1} \geq \lambda_{2} \geq \dots \geq \lambda_{n}$. $\lambda'_{i}$ is the corresponding eigenvalue after perturbation. While $\| \sum_{i = T + 1}^{n} \lambda_{i} \mathbf{u}_{i}\mathbf{u}_{i}^{T}X \|$ is hard to optimized, from~\eqref{equ.loss}, we can compute the upper bound instead of minimizing the loss directly. Accordingly, the goal of adversarial attack is to maximize the upper bound of the loss reversely. Thus the restrict black-box adversarial attack is equivalent to optimize:
\begin{align}
\notag\arg\max\limits_{A'} & \sum_{i = T + 1}^{n} {\lambda'_{i}}^{2} \cdot \sum_{i = T + 1}^{n} \|\mathbf{u}^{T}_{i}X\|_2^2,\\
\text{s.t.}& \; \| A' - A\| = 2\beta.
\label{equ.attackall}
\end{align}
Now our adversarial attack model is a general attacker. Theoretically, we can attack any graph embedding model which can be described by the corresponding graph filter $\mathscr{H}$. Meanwhile, our general attacker provides theoretical explanation on the transferability of adversarial samples created by \cite{KDD2018Adversarial,ICLR2019Meta,icml2019adversarial}, since modifying edges in adjacent matrix $A$ implicitly perturbs the eigenvalues of graph filters. In the following, we will analyze two kinds of popular graph embedding methods and aim to perform adversarial attack according to \eqref{equ.attackall}.
\subsection{\textit{GF-Attack} on Graph Convolutional Networks}
Graph Convolution Networks extend the definition of convolution to the irregular graph structure and learn a representation vector of a vertex with feature matrix $X$. Namely, we generalize the Fourier transform to graphs to define the convolution operation: $g_{\theta} \ast \mathbf{x} = U g_{\theta} U^{T} \mathbf{x}$. To accelerate calculation, ChebyNet \cite{Defferrard2016ChebNet} proposed a polynomial filter $g_\theta(\Lambda) = \sum_{k=0}^{K}\theta_k\Lambda^k$ and approximated $g_{\theta}(\Lambda)$ by a truncated expansion concerning Chebyshev polynomials $T_{k}(x)$:

\begin{equation}
g_{\theta'} \ast \mathbf{x} \approx \sum_{k = 0}^{K} \theta'_{k}T_{k}(\widetilde{L})\mathbf{x},
\end{equation}
where $\widetilde{L} = \frac{2}{\lambda_{\text{max}}}L - I_{n}$ and $\lambda_{\text{max}}$ is the largest eigenvalue of Laplacian matrix $L$. $\theta{'} \in \mathbb{R}^{K}$ is now the parameter of Chebyshev polynomials $T_{k}(x)$. $K$ denotes the $K_{\text{th}}$ order polynomial in Laplacian. Due to the natural connection between Fourier transform and single processing, it's easy to formulate ChebyNet to \textit{GF-Attack}:
\begin{lemma}
The $K$-localized single-layer ChebyNet with activation function $\sigma(\cdot)$ and weight matrix $\Theta$ is equivalent to filter graph signal $X$ with a polynomial filter $\mathscr{H} = \sum_{k=0}^{K}T_k(S)$ with graph-shift filter $S = 2\frac{L^{sym}}{\lambda_{max}} - I_n$. $T_k(S)$ represents Chebyshev polynomial of order $k$. Equation~\eqref{equ.GF-Attack} can be rewritten as:
\begin{equation*}
    \notag\tilde{X} = \sum_{k=0}^{K}T_k(2\frac{L^{sym}}{\lambda_{max}} - I_n)X,
    \,\,\,\, X' =  \sigma(\tilde{X}\Theta).
\end{equation*}
\label{lemma.chebynet}
\end{lemma}
\begin{proof}
The $K$-localized single-layer ChebyNet with activation function $\sigma(\cdot)$ is $\sigma(\sum_{k=0}^{K}\theta'_{k}T_k(2\frac{L^{sym}}{\lambda_{max}} - I_n)X)$. Thus, we can directly write graph-shift filter as $S = 2\frac{L^{sym}}{\lambda_{max}} - I_n$ and linear and shift-invariant filter $\mathscr{H} = \sum_{k=0}^{K}T_k(S)$.
\end{proof}

GCN \cite{ICLR2017SemiGCN} constructed the layer-wise model by simplifying the ChebyNet with $K=1$ and the \emph{re-normalization trick} to avoid gradient exploding/vanishing:
\begin{eqnarray}
\label{Eq:gcn}
X^{(l+1)} &=& \sigma\left(\tilde{D}^{-\frac{1}{2}}\tilde{A} \tilde{D}^{-\frac{1}{2}}X^{(l)}\Theta^{(l)}\right),
\end{eqnarray}
where $\tilde{A} = A + I_n$ and $\tilde{D}_{ii} = \sum_{j}\tilde{A}_{ij}$. $\Theta=\{\theta^{(l)}_{1},...,\theta^{(l)}_{n} \}$ is the parameters in the $l_{th}$ layer and $\sigma(\cdot)$ is an activation function.

SGC \cite{sgc_icml19} further utilized a single linear transformation to achieve computationally efficient graph convolution, i.e., $\sigma(\cdot)$ in SGC is a linear activation function. We can formulate the multi-layer SGC as \textit{GF-Attack} through its theoretical connection to ChebyNet:
\begin{corollary}
The $K$-layer SGC is equivalent to the $K$-localized single-layer ChebyNet with $K_{th}$ order polynomials of the graph-shift filter $S^{sym}= 2I_n - L^{sym}$. Equation~\eqref{equ.GF-Attack} can be rewritten as:
\begin{equation*}
    \notag\tilde{X} = (2I_n - L^{sym})^{K}X, \,\,\,\, X' =  \sigma(\tilde{X}\Theta).
\end{equation*}
\label{thm.sgc}
\end{corollary}
\vspace{-8mm}
\begin{proof}
We can write the $K$-layer SGC as $(2I_n - L^{sym})^{K} X \Theta$. Since $\Theta$ is the learned parameters by the neural network, we can employ the reparameterization trick to use $(2I_{n} - L^{sym})^{K}$ to approximate the same order polynomials $\sum_{k=0}^{K}T_k(2I_n - L^{sym})$ with new $\widetilde{\Theta}$. Then we rewrite the $K$-layer SGC by polynomial expansion as $\sum_{k=0}^{K}T_k(2I_n - L_{sym}) X \widetilde{\Theta}$. Therefore, we can directly write the graph-shift filter $S^{sym} = 2I_n - L^{sym}$ with the same linear and shift-invariant filter $\mathscr{H}$ as $K$-localized single-layer ChebyNet.
\end{proof}
Note that SGC and GCN are identical when $K=1$.
Even though non-linearity disturbs the explicit expression of graph-shift filter of multi-layer GCN, the spectral analysis from \cite{sgc_icml19} demonstrated that both GCN and SGC share similar graph filtering behavior. Thus practically, we extend the general attack loss from multi-layer SGC to multi-layer GCN under non-linear activation functions scenario. Our experiments also validate that the attack model for multi-layer SGC also shows excellent performance on multi-layer GCN.

\textbf{\textit{GF-Attack} loss for SGC/GCN.} As stated in Corollary~\ref{thm.sgc}, the graph-shift filter $S$ of SGC/GCN is defined as $S^{sym} =  2I_n - L^{sym} = D^{-\frac{1}{2}}AD^{-\frac{1}{2}} + I_n = \hat{A} + I_n$, where $\hat{A}$ denotes the normalized adjacent matrix. Thus, for $K$-layer SGC/GCN, we can decompose the graph filter $\mathscr{H}$ as $\mathscr{H}^{sym} = (S^{sym})^K = U_{\hat{A}} (\Lambda_{\hat{A}} + I_{n})^{K} U_{\hat{A}}^{T}$, where $\Lambda_{\hat{A}}$ and $U_{\hat{A}}$ are eigen-pairs of $\hat{A}$. The corresponding adversarial attack loss for $K_{th}$ order SGC/GCN can be rewritten as:
\vspace{-1.5ex}
\begin{equation}
\arg\max\limits_{A'} \sum_{i = T + 1}^{n} (\lambda'_{\hat{A'},i} + 1)^{2K} \cdot \sum_{i = T + 1}^{n} \|\mathbf{u}^{T}_{\hat{A'},i}X\|_2^2,
\label{equ.GF-Attack-sym}
\end{equation}
where $\lambda'_{\hat{A'},i}$ refers to the $i_{th}$ largest eigenvalue of the perturbed normalized adjacent matrix $\hat{A'}$.

While each time directly calculating $\lambda'_{\hat{A'},i}$ from attacked normalized adjacent matrix $A'$ will need an eigen-decomposition operation, which is extremely time consuming, eigenvalue perturbation theory is introduced to estimate $\lambda'_{\hat{A'},i}$ in a linear time:

\begin{theorem}\label{thm:General_Eigen}
Let $A' = A + \Delta A$ be a perturbed version of $A$ by adding/removing edges and $\Delta D$ be the respective change in the degree matrix. $\lambda_{\hat{A},i}$ and $\mathbf{u}_{\hat{A},i}$ are the $i_{th}$ eigen-pair of eigenvalue and eigenvector of $\hat{A}$ and also solve the generalized eigen-problem $A\mathbf{u}_{\hat{A},i}=\lambda_{\hat{A},i} D\mathbf{u}_{\hat{A},i}$. Then the perturbed generalized eigenvalue $\lambda^{'}_{\hat{A},i}$ is approximately as:
\begin{align}
    \lambda'_{\hat{A'},i} \approx \lambda_{\hat{A},i} + \frac{ \mathbf{u}^{T}_{\hat{A},i}\Delta A\mathbf{u}_{\hat{A},i} - \lambda_{\hat{A},i}\mathbf{u}^{T}_{\hat{A},i}\Delta Du_{\hat{A},i} }{\mathbf{u}^{T}_{\hat{A},i} D \mathbf{u}_{\hat{A},i}}.
\label{equ:General_Eigen}
\end{align} 
\end{theorem}
\begin{proof}
Please kindly refer to~\cite{zhu2018high}.
\end{proof}
With Theorem \ref{thm:General_Eigen}, we can directly derive the explicit formulation of $\lambda'_{\hat{A'}}$ perturbed by $\Delta A$ on adjacent matrix $A$.

\subsection{\textit{GF-Attack} on Sampling-based Graph Embedding}
Sampling-based graph embedding learns vertex representations according to sampled vertices, vertex sequences, or network motifs. For instance, LINE \cite{WWW2015Line} with second order proximity intends to learn two graph representation matrices $X'$, $Y'$ by maximizing the NEG loss of the skip-gram model:
\begin{equation}
\mathcal{L} = \sum_{i=1}^{|\mathcal{V}|} \sum_{j=1}^{|\mathcal{V}|} A_{i,j} \Big(\log \sigma(x'^{T}_{i} y'_{j}) + b\mathbb{E}_{j' \sim P_{N}}[\log \sigma(-x'^{T}_{i} y'_{j})] \Big),
\end{equation}
where $x'_{i}$, $y'_{i}$ are rows of $X'$, $Y'$ respectively; $\sigma$ is the sigmoid function; $b$ is the negative sampling parameter; $P_{N}$ denotes the noise distribution generating negative samples. Meanwhile, DeepWalk \cite{perozzi2014deepwalk} adopts the similar loss function except that $A_{i,j}$ is replaced with an indicator function indicating whether vertices $v_i$ and $v_j$ are sampled in the same sequence within given context window size $K$.

From the perspective of sampling-based graph embedding models, the embedded matrix is obtained by generating training corpus for the skip-gram model from adjacent matrix or a set of random walks. \cite{yang2015Comprehend,WSDM2018NetworkEmbedding} show that Point-wise Mutual Information (PMI) matrices are implicitly factorized in sampling-based embedding approaches. It indicates that LINE/DeepWalk can be rewritten into a matrix factorization form:
\begin{lemma}\label{lemma.deepwalk}\cite{WSDM2018NetworkEmbedding}
Given context window size $K$ and number of negative sampling $b$ in skip-gram, the result of DeepWalk in matrix form is equivalent to factorize matrix:
\vspace{-1ex}
\begin{equation}
M = \log{\Big(\frac{\text{vol}(G)}{bK}(\sum_{k=1}^{K}(D^{-1}A)^{k}){D}^{-1}\Big)},
\label{equ.deepwalk}
\end{equation}
where $\text{vol}(G) = \sum_{i}\sum_{j}A_{ij} = \sum_{i}D_{ii}$ denotes the volume of graph $G$. And LINE can be viewed as the special case of DeepWalk with $K=1$.
\end{lemma}

For proof of Lemma~\ref{lemma.deepwalk}, please kindly refer to \cite{WSDM2018NetworkEmbedding}.
Inspired by this insight, we prove that LINE can be viewed from a GSP manner as well:
\begin{theorem}
LINE is equivalent to filter a graph signal $X = \frac{1}{b}I_{n}$ with a polynomial filter $\mathscr{H}$ and fixed parameters $\Theta=\text{vol}(G)D^{-1}$. $\mathscr{H}=S$ is constructed by graph-shift filter $S^{rw}=I_n - L^{rw}$. Equation~\eqref{equ.GF-Attack} can be rewritten as:
\begin{align}
    \notag\tilde{X} &= \frac{1}{b}(I_{n} - L^{rw})D^{-1}I_{n}, \,\,\,\,
    \notag X'         = log(\text{vol}(G)\tilde{X}).
\end{align}
\label{thm.line}
\vspace{-3mm}
\end{theorem}
Note that LINE is formulated from an optimized unsupervised NEG loss of skip-gram model. 
Thus, the parameter $\Theta$ and value of the NCG loss have been fixed at the optimal point of the model with given graph signals.

We can extend Theorem~\ref{thm.line} to DeepWalk since LINE is a $1$-window special case of DeepWalk:
\begin{corollary}
The output of $K$-window DeepWalk with $b$ negative samples is equivalent to filtering a set of graph signals $X = \frac{1}{b}I_{n}$ with given parameters $\Theta=\text{vol}(G)D^{-1}$. Equation~\eqref{equ.GF-Attack} can be rewritten as:
\vspace{-2ex}
\begin{equation}
    \notag\tilde{X} = \frac{1}{bK}\sum_{k=1}^{K}(I_{n} - L^{rw})^{k}D^{-1}I_{n}, \,\,\,\,
    \notag X'         = log(\text{vol}(G)\tilde{X}). 
\end{equation}
\vspace{-6mm}
\label{col.deepwalk}
\end{corollary}

\begin{proof}[Proof of Theorem \ref{thm.line} and Corollary \ref{col.deepwalk}]
With Lemma~\ref{lemma.deepwalk}, we can explicitly write DeepWalk as
$\exp{(M)} = \frac{\text{vol}(G)}{b}(\sum_{k=1}^{K} \frac{1}{K}(I_n - L^{rw})^{k}D^{-1}I_n)$. Therefore, we can directly have the explicit expression of Equation~\eqref{equ.GF-Attack} on LINE/DeepWalk.
\end{proof}

\textbf{\textit{GF-Attack} loss for LINE/DeepWalk.} As stated in Corollary~\ref{col.deepwalk}, the graph-shift filter $S$ of DeepWalk is defined as $S^{rw} =  I_n - L^{rw} = D^{-1}A =  D^{-\frac{1}{2}}\hat{A}D^{\frac{1}{2}}$. Therefore, graph filter $\mathscr{H}$ of the $K$-window DeepWalk can be decomposed as $\mathscr{H}^{rw} = \frac{1}{K}\sum_{k=1}^{K}(S^{rw})^{k}$, which satisfies $ \mathscr{H}^{rw} D^{-1}= D^{-\frac{1}{2}}U_{\hat{A}}(\frac{1}{K}\sum_{k=1}^{K}\Lambda_{\hat{A}}^k)U_{\hat{A}}^{T}D^{-\frac{1}{2}}$.

Since multiplying $D^{-\frac{1}{2}}$ in \textit{GF-Attack} loss brings extra complexity, \cite{WSDM2018NetworkEmbedding} provides us a way to well approximate the perturbed $\lambda'_{\mathscr{H}^{rw}D^{-1}}$ without this term.

Inspired by \cite{WSDM2018NetworkEmbedding},
we can find that both the magnitude of eigenvalues and smallest eigenvalue of $\mathscr{H}^{rw}D^{-1}$ are always well bounded. Thus we can approximate $\lambda'_{\mathscr{H}^{rw}D^{-1}} \approx \frac{1}{d_{\min}}\lambda'_{U_{\hat{A}}(\frac{1}{K}\sum_{k=1}^{K}\Lambda_{\hat{A}}^k)U_{\hat{A}}^{T}}$. Therefore, the corresponding adversarial attack loss of  $K_{th}$ order DeepWalk can be rewritten as:
\begin{equation}
\arg\max\limits_{A'} \sum_{i = T + 1}^{n} (\frac{1}{d_{\min}}|\frac{1}{K}\sum_{k = 1}^{K}\lambda'^{k}_{\hat{A'},i}|)^{2} \cdot \sum_{i = T + 1}^{n} \|\mathbf{u}^{T}_{\hat{A'},i}X\|_2^2.
\label{equ.GF-Attack-rw}
\end{equation}
When $K=1$, Equation~\eqref{equ.GF-Attack-rw} becomes the adversarial attack loss of LINE. Similarly, Theorem \ref{thm:General_Eigen} is utilized to estimate $\lambda'_{\hat{A'}}$ in the loss of LINE/DeepWalk.

\begin{algorithm}[!tb]
\caption{Graph Filter Attack (GF-Attack) adversarial attack algorithm under RBA setting} \label{alg:Framework}
\small{
\begin{algorithmic}[1] 

\REQUIRE ~~\\ 
Adjacent Matrix $A$; 
feature matrix $X$; 
target vertex $t$; \\
number of top-$T$ smallest singular values/vectors selected $T$; order of graph filter $K$;
fixed budget $\beta$.\\

\ENSURE ~~\\ 
Perturbed adjacent Matrix $A'$.

\STATE Initial the candidate flips set as $\mathcal{C} = \{(v, t)|v \neq t \}$, eigenvalue decomposition of $\hat{A} = U_{\hat{A}}\Lambda_{\hat{A}}U_{\hat{A}}^{T}$;

\FOR {$(v, t) \in \mathcal{C}$}

\STATE Approximate $\Lambda'_{\hat{A}}$ resulting by removing/inserting edge $(v, t)$ via Equation~\eqref{equ:General_Eigen};
\STATE {Update ${Score_{(v, t)}}$ from loss Equation~\eqref{equ.GF-Attack-sym} or Equation~\eqref{equ.GF-Attack-rw};}

\ENDFOR
\STATE $\mathcal{C}_{sel}$ $\gets$ edge flips with top-$\beta$ $Score$;
\STATE $A' \gets A \pm \mathcal{C}_{sel}$;
\RETURN{$A'$}
\end{algorithmic}
}
\end{algorithm}

\subsection{The Attack Algorithm}
Now the general attack loss is established, the goal of our adversarial attack is to misclassify a target vertex $t$ from an attributed graph $G(\mathcal{V},\mathcal{E})$ given a downstream node classification task. We start by defining the candidate flips then the general attack loss is responsible for scoring the candidates.

We first adopt the hierarchical strategy in \cite{ICML2018Adversarial} to decompose the single edge selection into two ends of this edge in practice. Then we let the candidate set $\mathcal{C}$ for edge selection contains all vertices (edges and non-edges) directly accessary to the target vertex, i.e. $\mathcal{C} = \{(v, t)|v \neq t \}$, as \cite{ICML2018Adversarial,icml2019adversarial}. Intuitively, further away the vertices from target $t$, less influence they impose on $t$.
Meanwhile, experiments in \cite{KDD2018Adversarial,icml2019adversarial} also showed that they can do significantly more damage compared to candidate flips chosen from other parts of graph. Thus, our experiments are restricted on such choices.

Overall, for a given target vertex $t$, we establish the target attack by sequentially calculating the corresponding \textit{GF-Attack} loss w.r.t graph-shift filter $S$ for each flip in candidate set as scores. Then with a fixed budget $\beta$, the adversarial attack is accomplished by selecting flips with top-$\beta$ scores as perturbations on the adjacent matrix $A$ of clean graph. Details of the \textit{GF-Attack} adversarial attack algorithm under RBA setting is in Algorithm~\ref{alg:Framework}.

\begin{table*}[!t]
\centering
\caption{Summary of the change in classification accuracy (in percent) compared to the clean/original graph. Single edge perturbation under RBA setting. Lower is better. \label{tab:results single edge}}
\resizebox{\textwidth}{!}{%
\begin{tabular}{ l c c c c c c c c c c c c }
\toprule
    Dataset & \multicolumn{4}{c}{Cora} & \multicolumn{4}{c}{Citeseer} & \multicolumn{4}{c}{Pubmed}  \\
\cmidrule(lr){2-5}\cmidrule(l){6-9}\cmidrule(l){10-13}
Models & GCN & SGC & DeepWalk & LINE & GCN & SGC & DeepWalk & LINE & GCN & SGC & DeepWalk & LINE\\
(unattacked) &    80.20 & 78.82  & 77.23 & 76.75	& 72.50 & 69.68  & 69.68 & 65.15 		& 80.40 & 80.21	& 78.69 & 72.12 \\
\hline
\textit{Random} & -1.90 & -1.22  & -1.76 & -1.84	& -2.86 & -1.47  & -6.62 & -1.78 		& -1.75 & -1.77	& -1.25 & -1.01   \\
\textit{Degree} & -2.21 & -4.42  & -3.08 & -12.40	& -4.68 & -5.21  & -9.67 & -12.55 	& -3.86 & -4.44	& -2.43 & -13.05  \\
\textit{RL-S2V} & -5.20 & -5.62  & -5.24 & -10.38	& -6.50 & -4.08  & -12.13 & -20.10	& -6.40 & -6.11 & -6.10 & -13.21\\
\textit{$\mathcal{A}_{class}$} & -3.62 & -2.96  & \textbf{-6.29} & -7.55	& -3.48 & -2.83  & \textbf{-12.56} & -10.28	& -4.21 & -2.25 & -3.05 & -6.75 \\
\midrule
\textit{GF-Attack} &  \textbf{-7.60}  & \textbf{-9.73} & -5.31 & \textbf{-13.27} & \textbf{-7.78} & \textbf{-6.19} & -12.50 & \textbf{-22.11} & \textbf{-7.96} & \textbf{-7.20} & \textbf{-7.43} & \textbf{-14.16}\\
\bottomrule
\end{tabular}
}
\end{table*}

\section{Experiments}\label{sec.exp}
\textbf{Datasets.}
We evaluate our approach on three real-world datasets: Cora \cite{Dataset2000Cora}, Citeseer and Pubmed \cite{Dataset2008Citeseer}. In all three citation network datasets, vertices are documents with corresponding bag-of-words features and edges are citation links. The data preprocessing settings are closely followed the benchmark setup in \cite{ICLR2017SemiGCN}.
Only the largest connected component (LCC) is considered to be consistent with~\cite{KDD2018Adversarial}.
For statistical overview of datasets, please kindly refer to~\cite{KDD2018Adversarial}. Code and datasets are available at https://github.com/SwiftieH/GFAttack.



\textbf{Baselines.}
In current literatures, few of studies strictly follow the restricted black-box attack setting. They utilize the additional information to help construct the attackers, such as labels \cite{KDD2018Adversarial}, gradients \cite{ICML2018Adversarial}, etc. 

Hence, we compare four baselines with the proposed attacker under RBA setting as follows:
\begin{itemize}[noitemsep,topsep=0pt,parsep=0pt,partopsep=0pt]
\item \textit{Random} \cite{ICML2018Adversarial}:  for each perturbation, randomly choosing insertion or removing of an edge in graph $G$. We report averages over 10 different seeds to alleviate the influence of randomness.
\item \textit{Degree} \cite{CIKM2012Gelling}:  for each perturbation, inserting or removing an edge based on degree centrality, which is equivalent to the sum of degrees in original graph $G$.
\item \textit{RL-S2V} \cite{ICML2018Adversarial}: a reinforcement learning based attack method, which learns the generalizable attack policy for GCN under RBA scenario.
\item \textit{$\mathcal{A}_{class}$} \cite{icml2019adversarial}: a matrix perturbation theory based black-box attack method designed for DeepWalk. Then \textit{$\mathcal{A}_{class}$} evaluates the targeted 
attacks on node classification by learning a logistic regression.
\end{itemize}

\textbf{Target Models.}
To validate the generalization ability of our proposed attacker, we choose four popular graph embedding models: GCN \cite{ICLR2017SemiGCN}, SGC \cite{sgc_icml19}, DeepWalk \cite{perozzi2014deepwalk} and LINE \cite{WWW2015Line} for evaluation. First two of them are Graph Convolutional Networks and the others are sampling-based graph embedding methods. For DeepWalk, the hyperparameters are set to commonly used values: window size as $5$, number of negative sampling in skip-gram as $5$ and top-$128$ largest singular values/vectors. A logistic regression classifier is connected to the output embeddings of sampling-based methods for classification. Unless otherwise stated, all Graph Convolutional Networks contain two layers. 

\textbf{Attack Configuration.}
A small budget $\beta$ is applied to regulate all the attackers. To make this attacking task more challenging, $\beta$ is set to 1. Specifically, the attacker is limited to only add/delete a single edge given a target vertex $t$. For our method, we set the parameter $T$ in our general attack model as $n - T = 128$, which means that we choose the top-$T$ smallest eigenvalues for $T$-rank approximation in embedding quality measure. Unless otherwise indicated, the order of graph filter in \textit{GF-Attack} model is set to $K=2$. Following the setting in \cite{KDD2018Adversarial}, we split the graph into labeled (20\%) and unlabeled vertices (80\%). Further, the labeled vertices are splitted into equal parts for training and validation. The labels and classifier is invisible to the attacker due to the RBA setting. The attack performance is evaluated by the decrease of node classification accuracy following \cite{ICML2018Adversarial}.

\begin{figure*}[!t]
\centering
\subfigure {\includegraphics[width=0.24\linewidth]{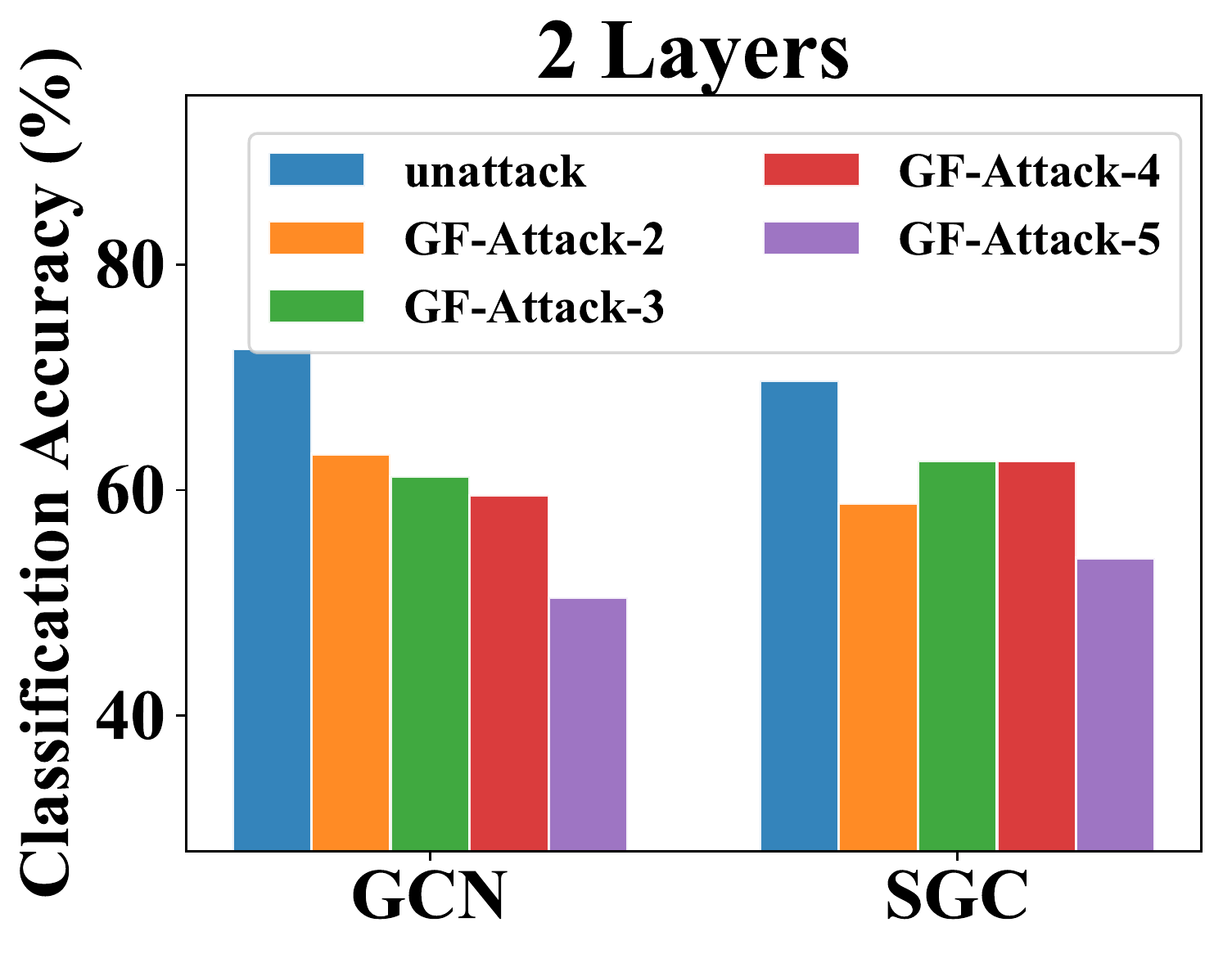}}
\subfigure {\includegraphics[width=0.24\linewidth]{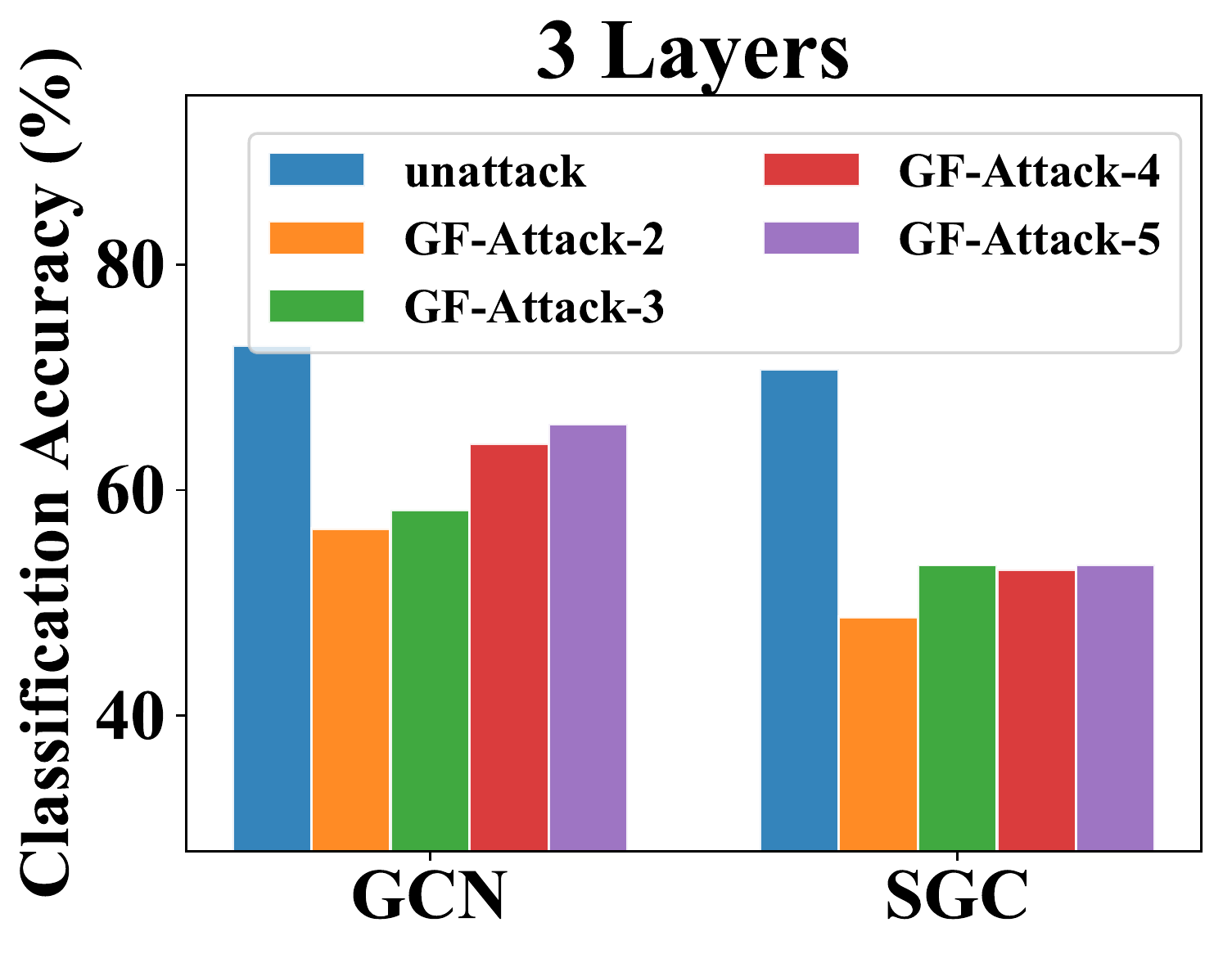}}
\subfigure {\includegraphics[width=0.24\linewidth]{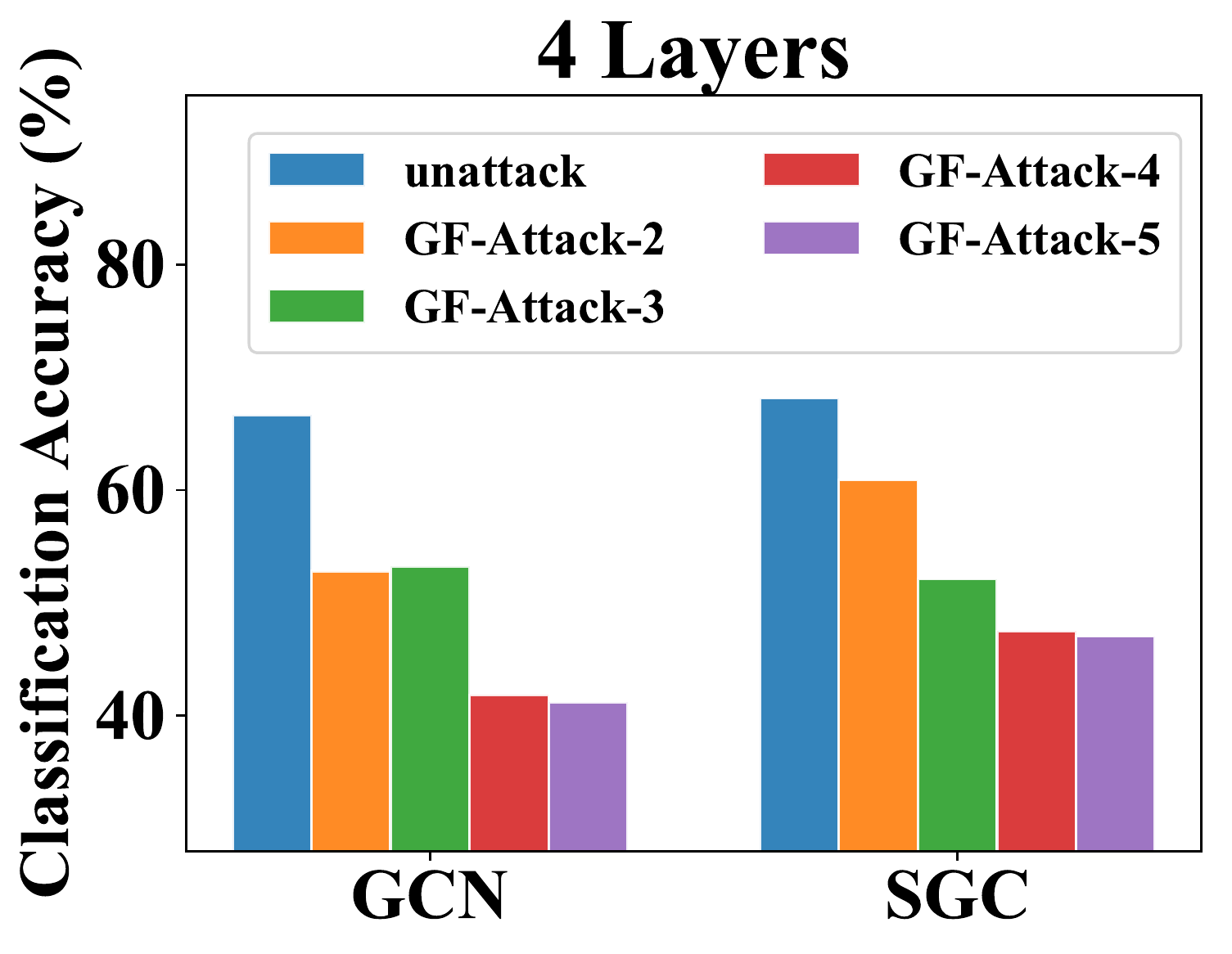}}
\subfigure {\includegraphics[width=0.24\linewidth]{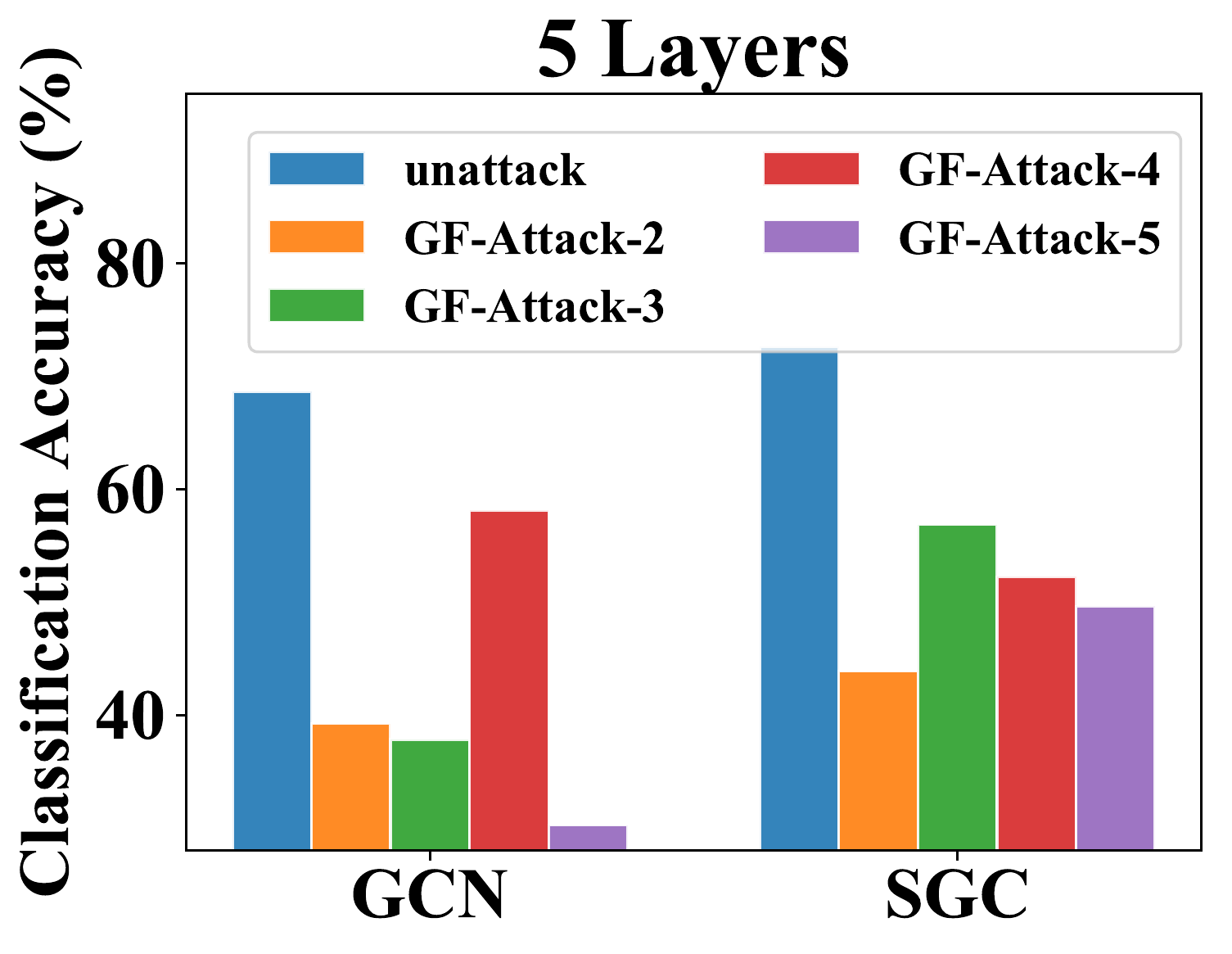}}
\vspace{-2mm}
\caption{Comparison between order $K$ of \textit{GF-Attack} and number of layers in GCN/SGC on Citeseer.}
	\label{fig:layer vs order}
\end{figure*}

\begin{table}[htbp]
\renewcommand{\arraystretch}{1.11}
\caption{Running time ($s$) comparison over all baseline methods on Citeseer. We report the $10$ times average running time of processing single node for each model.} \label{tab:timeCost}
\vspace{-2mm}
\begin{center}
\resizebox{\columnwidth}{!}{%
\begin{tabular}{ c c c c c c }
\hline
Models  & \textit{Random} & \textit{Degree} & \textit{RL-S2V} & \textit{$\mathcal{A}_{class}$} & \textit{GF-Attack} \\
\hline
Citeseer    & 0.19      & 42.21     & 222.80  &146.58 & 12.78   \\
\hline
\end{tabular}
}
\end{center}
\end{table}

\subsection{Attack Performance Evaluation}
In the section, we evaluate the overall attack performance of different attackers.

\textbf{Attack on Graph Convolutional Networks.}
Table \ref{tab:results single edge} summaries the attack results of different attackers on Graph Convolutional Networks. Our \textit{GF-Attack} attacker outperforms other attackers on all datasets and all models. Moreover, \textit{GF-Attack} performs quite well on 2 layers GCN with nonlinear activation. This implies the generalization ability of our attacker on Graph Convolutional Networks. 

\textbf{Attack on Sampling-based Graph Embedding.}
Table \ref{tab:results single edge} also summaries the attack results of different attackers on sampling-based graph embedding models. As expected, our attacker achieves the best performance nearly on all target models. It validates the effectiveness of our method on attacking sampling-based models. 

Another interesting observation is that the attack performance on LINE is much better than that on DeepWalk. This result may due to the deterministic structure of LINE, while the random sampling procedure in DeepWalk may help raise the resistance to adversarial attack. Moreover, \textit{GF-Attack} on all graph filters successfully drop the classification accuracy on both Graph Convolutional Networks and sampling-based models, which again indicates the transferability of our general model in practice.

\subsection{Evaluation of Multi-layer GCNs}
To further inspect the transferability of our attacker, we conduct attack towards multi-layer Graph Convolutional Networks w.r.t the order of graph filter in \textit{GF-Attack} model. 
Figure~\ref{fig:layer vs order} presents the attacking results on $2$, $3$, $4$ and $5$ layers GCN and SGC with different orders, and the number followed by \textit{GF-Attack} indicates the graph-shift filter order $K$ in general attack loss. 
From Figure~\ref{fig:layer vs order}, we can observe that: first, the transferability of our general model is demonstrated, since all graph-shift filters in loss with different order $K$ can perform the effective attack on all models. Interestingly, \textit{GF-Attack-5} achieves the best attacking performance in most cases. It implies that the higher order filter contains higher order information and has positive effects on attack to simpler models. Second, the attacking performance on SGC is always better than GCN under all settings. We conjecture that the non-linearity between layers in GCN successively adding robustness to GCN.

\subsection{Evaluation under Multi-edge Perturbation Settings}
In this section, we evaluate the performance of attackers with multi-edge perturbation, i.e. $\beta \geq 1$. 
The results of multi-edge perturbations on Cora under RBA setting are reported in Figure~\ref{fig:multi-edge perturbation} for demonstration.
Clearly, with increasing of the number of perturbed edges, the attacking performance gets better for each attacker. Our attacker outperforms other baselines on all cases. It validates that our general attacker can still perform well when fixed budget $\beta$ becomes larger.

\begin{figure}[htbp]
\centering
\subfigure [small][GCN]{\includegraphics[width=0.48\linewidth]{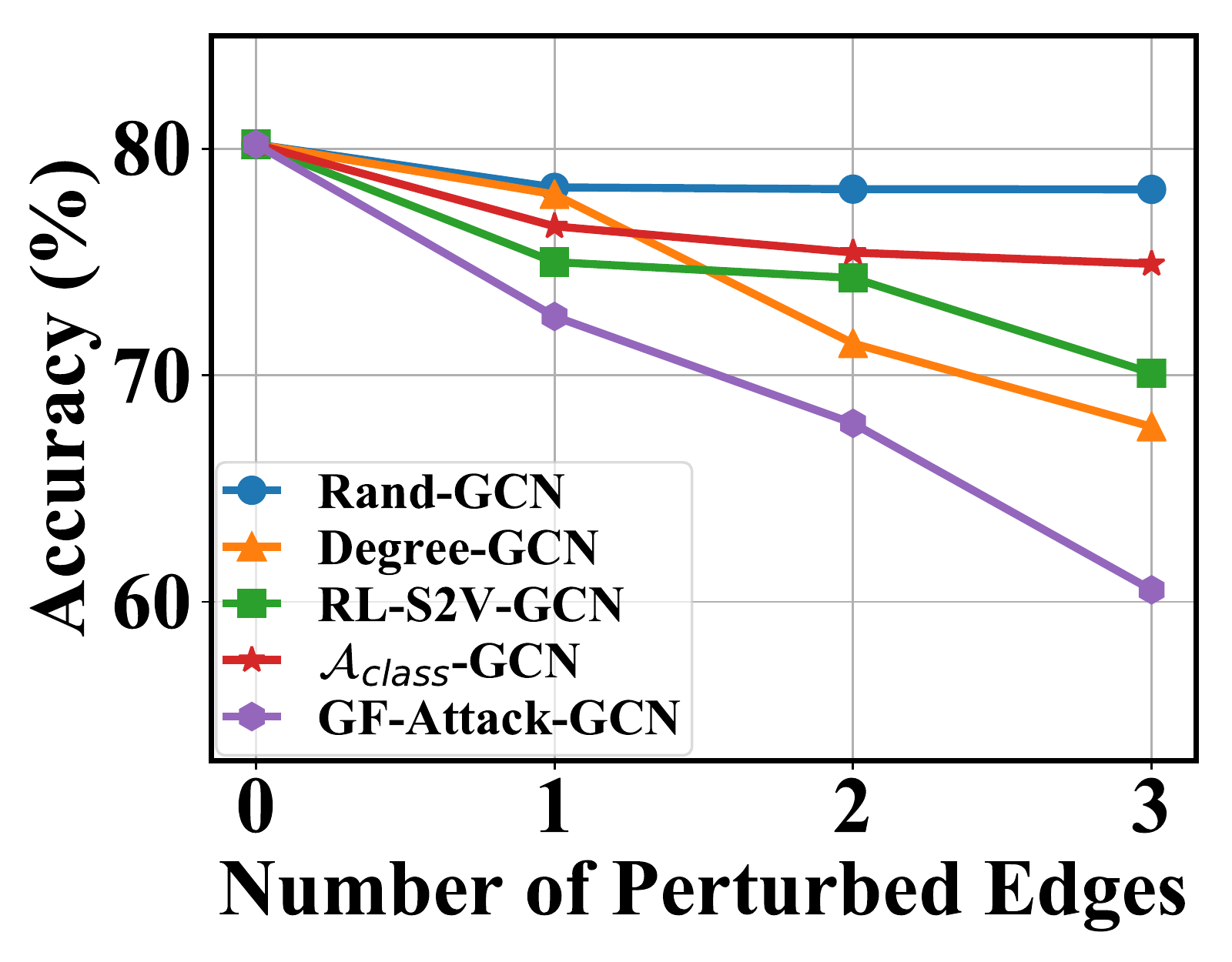}}
\subfigure [small][SGC]{\includegraphics[width=0.48\linewidth]{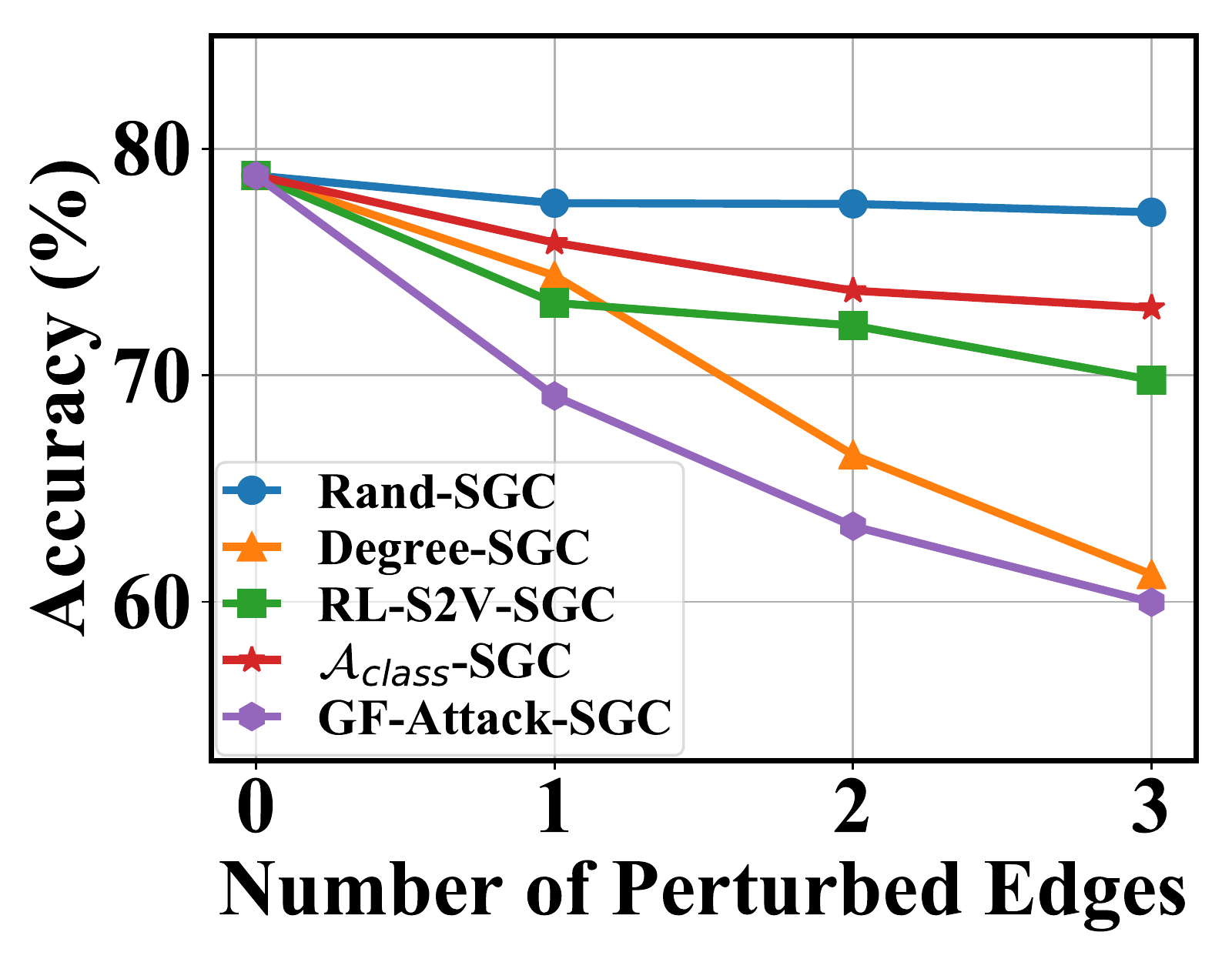}}
\caption{Multiple-edge attack results on Cora under RBA setting. Lower is better.}
\label{fig:multi-edge perturbation}
\end{figure}

\subsection{Computational Efficiency Analysis}
In this section, we empirically evaluate the computational efficiency of our \textit{GF-Attack}. The running time ($s$) comparison of $10$ times average on Citeseer is demonstrated in Table~\ref{tab:timeCost}. While being less efficient than two native baselines (\textit{Random} and \textit{Degree}), our \textit{GF-Attack} is much faster than the developed methods \textit{RL-S2V} and \textit{$\mathcal{A}_{class}$}. Joining the performance in Table~\ref{tab:results single edge}, it reads that \textit{GF-Attack} is not only effective in performance but also efficient computationally.

\section{Conclusion}\label{sec.conclusion}
In this paper, we consider the adversarial attack on different kinds of graph embedding models under restrict black-box attack scenario. From graph signal processing of view, we try to formulate the graph embeddding method as a general graph signal process with corresponding graph filter and construct a restricted adversarial attacker which aims to attack the graph filter only by the adjacency matrix and feature matrix. Thereby, a general optimization problem is constructed by measuring embedding quality and an effective algorithm is derived accordingly to solve it. Experiments show the vulnerability of different kinds of graph embedding models to our attack framework.

\section{Acknowledgements}\label{sec.acknowledgment}

This work is supported by National Natural Science Foundation of China Major Project No. U1611461 and National Program on Key Basic Research Project No. 2015CB352300. We would like to thank the anonymous reviewers for the helpful comments. We also thank Daniel Zügner from Technical University of Munich for the valuable suggestions and discussions.

\bibliographystyle{aaai}
\bibliography{GF_Attack}

\end{document}